\documentclass[12pt,a4paper]{article}
\usepackage[utf8]{inputenc}
\usepackage{amssymb,amsfonts,amsmath,bm}
\usepackage{subfigure}
\usepackage{graphicx}
\usepackage{hyperref}
\usepackage{dsfont}
\usepackage{bbold}
\usepackage{amsthm}
\usepackage{color}
\usepackage{mathtools}
\usepackage{multicol}
\usepackage{tikz}
\usepackage{cite}
\usepackage{float}
\usepackage[bbgreekl]{mathbbol}
\usepackage{amsfonts}

\usepackage{amsmath}
\usetikzlibrary{positioning}
\allowdisplaybreaks

\usepackage[
top    = 3.2cm,
bottom = 3.0cm,
left   = 2.8cm,
right  = 2.8cm]{geometry}

\newtheorem{thm}{Theorem}[section]
\newtheorem{proposition}[thm]{Proposition}

\newtheorem{remark}[thm]{Remark}

\begin{document}

\begin{titlepage}

\begin{center}

{\Large \bf Reductions and degenerate limits of Yang-Baxter maps with $3\times 3$ Lax matrices}

\vskip 1.5cm

{{\bf P. Adamopoulou$^{\star}$, \bf T.E. Kouloukas$^{\dagger}$ and G. Papamikos$^{\ast}$ }} 

\vskip 0.8cm

{\footnotesize
$^{\star}$ Maxwell Institute for Mathematical Sciences and Department of Mathematics, \\ Heriot-Watt University}
\\
{\footnotesize
$^{\dagger}$ School of Computing and Digital Media, London Metropolitan University}
\\
{\footnotesize
$^{\ast}$ School of Mathematics, Statistics and Actuarial Science, University of Essex}

\vskip 0.5cm

{\footnotesize {\tt E-mail: p.adamopoulou@hw.ac.uk, t.kouloukas@londonmet.ac.uk  g.papamikos@essex.ac.uk }}\\
\end{center}

\vskip 2.0cm

\begin{abstract}
\noindent We generalise a family of quadrirational parametric Yang-Baxter maps with $3\times 3$ Lax matrices by introducing additional essential parameters. These maps preserve a prescribed Poisson structure which originates from the Sklyanin bracket. We investigate various low-dimensional reductions of this family, as well as degenerate limits with respect to the parameters that were introduced. As a result, we derive several birational Yang-Baxter maps, and we discuss some of their integrability properties. This work is part of a more general classification of Yang-Baxter maps admitting a strong $3\times 3$ Lax matrix with a linear dependence on the spectral parameter.
\end{abstract}

\hspace{.2cm} \textbf{Mathematics Subject Classification:} 16T25, 37J10, 14E05

\hspace{.2cm} \textbf{Keywords:} Yang-Baxter equation, birational maps, Lax matrices,

\hspace{.2cm} discrete dynamical systems, symplectic maps, Liouville integrability.

\vfill
\end{titlepage}

\section{Introduction}

\subsection{Yang-Baxter maps and Lax matrices}

In \cite{Drinfeld1992}, the study of set theoretical solutions to the Yang-Baxter (YB) equation \cite{Baxter, Yang}  was proposed, with such solutions now known as YB maps. A map $\mathbf{R}: \mathcal{X} \times
\mathcal{X} \to \mathcal{X} \times \mathcal{X}$, with $\mathcal{X}$ any set, is called a YB map  \cite{Bukhshtaber1998,Veselov2003} if it satisfies the set theoretical YB equation
\begin{equation}\label{YB eq}
 \mathbf{R}^{12}\circ \mathbf{R}^{13}\circ
\mathbf{R}^{23}= \mathbf{R}^{23}\circ \mathbf{R}^{13}\circ \mathbf{R}^{12} \,.
\end{equation}
$\mathbf{R}^{ij}$, for $i\neq j \in \{1,2,3\}$, denotes the action of the map $\mathbf{R}$ on the $i$ and $j$ factors of $\mathcal{X} \times \mathcal{X} \times \mathcal{X}$ and identically on the remaining factor, i.e. $\mathbf{R}^{12}=\mathbf{R} \times id_{\mathcal{X}}$ where $id_{\mathcal{X}}$ is the identity map over $\mathcal{X}$. While in general the set $\mathcal{X}$ can be any set, here we assume that it is an algebraic variety over a field $\mathbb{F}$ of characteristic zero. A parametric YB map \cite{Veselov2003, Veselov2007} is a YB map which depends on parameters $a, b \in \mathcal{P} \subset \mathbb{F}^d $, acting as $\mathbf{R}: (\mathcal{X} \times \mathcal{P}) \times (\mathcal{X} \times \mathcal{P}) \to (\mathcal{X} \times \mathcal{P}) \times (\mathcal{X} \times \mathcal{P})$, with
\begin{equation}
    \mathbf{R}((x,a),(y,b)) = \big((u(x,a,
y,b), a),(v(x,a,y,b),b) \big):=(u,v)\,.
\end{equation}
We will refer to the parameters $a,b \in \mathcal{P}$ as YB parameters and 
 denote such a YB map simply by $\mathbf{R}_{a,b}:(x,y) \mapsto (u,v)$, i.e. as a map from $\mathcal{X} \times \mathcal{X}$ to itself. The YB maps that we consider are all birational maps. We call a birational map $(x,y)\to(u,v)$ quadrirational \cite{ABS2004,PTSV2010} or non-degenerate if both maps $u(\cdot, y)$ and $v(x, \cdot)$ are birational isomorphisms of $\mathcal{X}$ to itself.

Of particular relevance are YB maps which arise from refactorisation problems of Lax matrices. A matrix $\mathbf{L}$ depending on $x \in \mathcal{X}$, $a \in \mathcal{P}$ and another parameter $\lambda \in \mathbb{F}$, is called a Lax matrix \cite{SurisVeselov2003, Veselov2003} for a parametric YB map $\mathbf{R}_{a,b}:(x,y) \mapsto (u,v)$, if 
\begin{equation} \label{refLx}
    \mathbf{L}(u,a, \lambda) \mathbf{L}(v,b, \lambda) = \mathbf{L}(y,b,\lambda) \mathbf{L}(x,a,\lambda)\;.
\end{equation}
If the refactorisation problem \eqref{refLx} is equivalent to  $(u,v)=\textbf{R}_{a,b}(x,y)$, then $\mathbf{L}$ is called  strong Lax matrix. The maps $\textbf{R}_{a,b}$ obtained in this way are all birational but not necessarily quadrirational. In addition, if the equation
\begin{equation}
\mathbf{L}(\hat{x},a,\lambda)\mathbf{L}(\hat{y},b,\lambda)\mathbf{L}(\hat{z},c,\lambda)=\mathbf{L}(x,a,\lambda)\mathbf{L}(y,b,\lambda)\mathbf{L}(z,c,\lambda)
\end{equation}
implies the unique solution $\hat{x}=x$, $\hat{y}=y$, $\hat{z}=z$, then it follows that the map $\textbf{R}_{a,b}$ is a YB map \cite{KP2009}.

Various YB maps and their generalisations in higher dimensions or over associative algebras (but not necessarily commutative) have been constructed recently, see \cite{AdamKRPap2021, BobenkoSuris2002, Doliwa2014, KasKoul2022, KNPT,KRM2013, MikPapWang2016} and references therein. In this work, we focus on refactorisation problems of $3 \times 3$ Lax matrices. Starting from an $18$-dimensional YB map, which we call principal parametric YB map, we derive lower-dimensional multi-parametric YB maps via several types of reductions. The obtained maps admit invariant quantities (first-integrals) that Poisson commute with respect to an $r$-matrix Poisson structure (Sklyanin bracket). After a degenerate limit is considered, the resulting maps lose their quadrirationality and become birational. These reduced birational maps can be thought of as vectorial and multi-parametric generalisations of the Adler-Yamilov (AY) map which is related to the nonlinear Schr\"odinger equation \cite{AdlerYamilov1994}.

\subsection{Poisson Yang-Baxter maps with binomial Lax matrices}

To construct multidimenional Yang-Baxter maps it is very natural to study
the solutions of the refactorisation problem, 
\begin{equation} \label{refbig}
\mathbf{L}(U,a,\lambda)\mathbf{L}(V,b,\lambda)=\mathbf{L}(Y,b,\lambda)\mathbf{L}(X,a,\lambda)
\end{equation}
with respect to $U$ and $V$ and with Lax matrices $\mathbf{L}(X,a,\lambda)$ which are first degree polynomials in the spectral parameter $\lambda$, 
\begin{equation}
\mathbf{L}(X,a,\lambda)=X-\lambda K_a.
\label{eq: Lax}
\end{equation}
Here, $X,Y,U,V$ are generic elements in $\mathfrak{gl}_n(\mathbb{F})$, while $K : \mathbb{F}^d \to \mathrm{GL}_n(\mathbb{F})$ is a $d$-parametric family of commuting
matrices and $K_a,K_b$ denote the values $K(a),K(b)$ respectively. In \cite{KP2009,KP2011}, solutions of this refactorisation problem were presented which satisfy the Yang-Baxter equation. We can express these solutions recursively as
\begin{equation}\label{solref}
    \begin{split}
        U &= \left( -f_0(X; a) I - \sum_{i=1}^{n} (-1)^i f_i(X; a) M_{i-1} \right) \left( \sum_{i=1}^{n} (-1)^i f_i(X; a) M_{i-1} \right)^{-1} K_{a}\;, \\
        V &= K_{a}^{-1} (Y K_{a} + K_{b} X - UK_{b})\;,
    \end{split}
\end{equation}
where $M_0 = I$, $N_0 = 0$, $M_1 = (Y K_{a} + K_{b} X) K_{b}^{-1}$, $N_1 = -Y K_{b}^{-1} K_{a}$ and 
\[
M_i = M_1 M_{i-1} + N_1 N_{i-1}, \quad N_i = N_1 M_{i-1}, \quad \text{for} \quad i = 2, \dots, n.
\]
Here, the functions $f_i$, for $i=0,\dots n$, are defined by the coefficients of the polynomial $p^a_{\lambda}(X)=\det(X-\lambda K_a)$, by the expression:
\[
p^a_{\lambda}(X) = (-1)^n f_n(X, a) \lambda^n + (-1)^{n-1} f_{n-1}(X, a) \lambda^{n-1} + \dots + (-1) f_1(X, a) \lambda + f_0(X, a)\;,
\]
with $f_n(X, a) = \det K_a$ and $f_0(X,a) = \det X$.

The solution \eqref{solref} satisfies the additional conditions 
\begin{equation} \label{condi}
f_i(U, a) = f_i(X, a), \  f_i(V, b) = f_i(Y, b), \ \ i = 0, \dots, n\;,
\end{equation}
or equivalently the condition $\det(U K_b-Y K_a)\neq 0$ (or $\det(K_a V-K_b Y)\neq 0$). 
The corresponding map $\mathcal{R}_{a,b}:(X,Y)\mapsto(U,V)$ defined by \eqref{solref}, is a quadrirational Yang-Baxter map. Furthermore, $\mathcal{R}_{a,b}$ is a Poisson map with respect to the Sklyanin bracket \cite{Sklyanin1982}
\begin{equation} \label{Skly}
    \lbrace \mathbf{L}(X,a,\lambda_1) \overset{\otimes}{,} \mathbf{L}(X,a,\lambda_2)  \rbrace = \left[   \frac{P}{\lambda_1-\lambda_2}, \mathbf{L}(X,a,\lambda_1) \otimes \mathbf{L}(X,a,\lambda_2) \right], 
\end{equation}
and $\lbrace \mathbf{L}(X,a,\lambda_1) \overset{\otimes}{,} \mathbf{L}(Y,b,\lambda_2)\rbrace=0$, where $P(x\otimes y)=y\otimes x$. 
The functions $f_i$, along with all elements of $K_a$ and $K_b$, are Casimirs for this Poisson bracket. Hence, the invariant conditions \eqref{condi} allow us to further reduce the $2n^2$-dimensional map $\mathcal{R}_{a,b}$ to a $2n(n-1)$-dimensional symplectic Yang-Baxter map on the level sets
\begin{equation}
\mathcal{C}=\lbrace (X,Y):f_i(X)=\alpha_i, f_i(Y)=\beta_i, \ i=0,\dots n-1 \rbrace \subset \mathfrak{gl}_n(\mathbb{F})\times\mathfrak{gl}_n(\mathbb{F})\;,
\label{casC}
\end{equation}
where $\alpha_i$ and $\beta_i$ represent additional YB parameters.

In this paper, we study the case where $K_a$ and $K_b$ are $3 \times3$ diagonal matrices. Our analysis covers all cases of binomial Lax matrices with diagonalisable higher-degree term,  as equation \eqref{refbig} remains invariant under conjugation with a constant matrix. We will investigate lower dimensional reductions and specific limits leading to non-quadrirational Yang-Baxter maps.

\section{Yang-Baxter maps with $3\times 3$ Lax matrices} \label{sec:3x3}

\subsection{The principal parametric Yang-Baxter map}
We consider the refactorisation problem \eqref{refbig} for generic $3\times 3$ matrices $X,Y,U,V$ and nonzero diagonal $3$-parametric matrices $K_a,K_b$. Using the scaling symmetry $\mathbf{L}\to r \mathbf{L}$, with $r \in \mathbb{F}\backslash \lbrace 0 \rbrace$, of equation \eqref{refbig} one can rescale any of the parameters of $K_a$ and $K_b$ to $1$, without loss of generality. Nevertheless, in this section we keep all parameters in $K_a, K_b$ arbitrary, as elements of a projective space, and we will use the rescaling when we consider certain reductions in later sections. 

We start with the general $3 \times 3$ Lax matrix of the form $\mathbf{L}(X,a,\lambda) = X- \lambda K_{a}$ with $\lambda \in \mathbb{F}$ and 
\begin{equation}\label{X, K mtrx}
    X= \begin{pmatrix}
        x_{11} & x_{12} & x_{13}\\
        x_{21} & x_{22} & x_{23} \\
        x_{31} & x_{32} & x_{33}
    \end{pmatrix},
     \quad
     K_a = \begin{pmatrix}
         a_1 &0&0\\
         0& a_2 &0\\
         0& 0 &a_3
     \end{pmatrix},
\end{equation}
where $X\in \mathfrak{gl}_3(\mathbb{F})$ and $a=(a_1,a_2,a_3)$ is an element of the $\mathbb{F}$-projective plane $\mathbb{P}^2(\mathbb{F})$. The Sklyanin bracket \eqref{Skly} implies the following linear Poisson bracket between the variables $x_{ij}$
\begin{equation}\label{PBs}
    \lbrace  x_{ij}, x_{kl} \rbrace = a_i x_{kj} \delta_{li} - a_j x_{il} \delta_{kj} \,,
\end{equation}
while $\lbrace x_{ij}, a_k \rbrace =0$.
This Poisson bracket admits six linearly independent Casimir functions $ a_1,a_2,a_3$, $f_0,f_1,f_2$, where $f_i$ are defined by the coefficients of the polynomial $p^a_{\lambda}(X)=\det(X-\lambda K_a)$, i.e., 
\begin{equation} \label{Cas}
    \begin{split}
        f_0(X,a) &= \det X, \\
        f_1(X,a) &= a_3(x_{11}x_{22} - x_{12}x_{21}) + a_2(x_{11}x_{33} - x_{13}x_{31}) + a_1( x_{22}x_{33} - x_{23}x_{32}), \\
        f_2(X,a) &= a_2a_3x_{11} + a_1a_3x_{22} + a_1a_2x_{33}, \\
        f_3(X,a) &=\det K_a\,,
    \end{split}
\end{equation}

In this case, the refactorisation problem \eqref{refbig}, implies uniquely an $18$-dimensional Poisson Yang-Baxter map  
$$\mathbf{R}_{a,b}:(X,Y)\mapsto(U,V)$$ 
defined by \eqref{solref}, which can be reduced to a $12$-dimensional symplectic YB map 
$\mathrm{R}_{\bar{a},\bar{b}}$ on $\mathcal{C}$, the intersection $\cap_{i=0}^2f_i^{-1}(\alpha_i)\times f_i^{-1}(\beta_i)\subset \mathfrak{gl}_3(\mathbb{F})\times\mathfrak{gl}_3(\mathbb{F})$, where $\alpha_i$ and $\beta_i$ are additional Yang-Baxter parameters taking values in $\mathbb{F}$. Here we denote by $\bar{a}=((a_1,a_2,a_3),(\alpha_0,\alpha_1,\alpha_2))\in \mathbb{P}^2(\mathbb{F})\times \mathbb{F}^3$ and similarly for $\bar{b}$. This parametric family of YB maps $\mathrm{R}_{\bar{a},\bar{b}}$, containing ten effective parameters, constitutes a generalisation of the family derived in Proposition 4.4 of \cite{KP2011} where the case $K_a=K_b=I$ was considered. In the following sections, we show that further lower dimensional reductions of this map are possible.

\subsection{Reduction to $8$-dimensional Yang-Baxter map}

The $9$-dimensional Poisson manifold $\mathcal{L}:=\{\mathbf{L}(X,a,\lambda): a \ \text{constant}\}$, equipped with the Sklyanin bracket \eqref{PBs}, has rank six. However, the rank of this Poisson structure can be reduced to four by imposing constraints on $x_{ij}$ that result in the vanishing of  matrix minors. We denote by $M_{ijk, lmn}$ the minor formed by deleting rows $i,j,k$ and columns $l,m,n$ from the matrix in \eqref{Skly}. We consider the minors 
\begin{equation}\label{minors}
    \begin{split}
        M_{789,125} &= -a_1a_2a_3^2 \left( a_2(x_{13}^2x_{21} -x_{13}x_{11}x_{23}) + a_1(x_{13}x_{22}x_{23} - x_{12}x_{23}^2)  \right)^2, \\
        M_{589,127} &= -a_1^2a_2a_3 \left( a_3(x_{12}^2x_{23} - x_{12}x_{13}x_{22}) + a_2(x_{13}x_{12}x_{33} - x_{13}^2x_{32}   \right)^2,\\ 
        M_{478,478} & = -a_1^2a_2^2a_3^2 \left( x_{12}x_{23}x_{31} - x_{13}x_{21}x_{32} \right)^2.
    \end{split}
\end{equation}
We notice that the system of equations $M_{789,125}=0,~ M_{589,127}=0,~M_{478,478}=0$ is linear in $x_{11},~x_{31},~x_{32}$, and for 
$x_{13}, x_{23} \neq 0$ implies the solution  
\begin{equation} \label{x11x31x32}
    \begin{split}
        x_{11} &= \frac{x_{13} x_{21}}{x_{23}} + \frac{a_1(x_{13}x_{22} - x_{12}x_{23})}{a_2x_{13}} \,,\\
        x_{31} &= \frac{x_{21}x_{33}}{x_{23}} - \frac{a_3}{a_2} \frac{x_{21}x_{22}}{x_{23}} + \frac{a_3}{a_2} \frac{x_{21} x_{12}}{x_{13}} \,, \\
        x_{32} &= \frac{x_{12}x_{33}}{x_{13}} - \frac{a_3}{a_2} \frac{x_{12}x_{22}}{x_{13}} + \frac{a_3}{a_2} \frac{x_{12}^2 x_{23}}{x_{13}^2}\,.
    \end{split}
\end{equation}
Substituting relations \eqref{x11x31x32} in the Casimir functions \eqref{Cas} we obtain the following set of reduced rational Casimirs $\iota^*f_i=f_i\circ \imath :\mathbb{F}^6\to\mathbb{F}$ 
\begin{equation} \label{casM}
    \begin{split}
        \iota^*f_0 &= \frac{(x_{13} x_{22}-x_{12} x_{23})^2 \left(a_2 x_{13} (a_1 x_{23} x_{33}+ a_3 x_{13} x_{21})+ a_1 a_3 x_{12} x_{23}^2\right)}{a_2^2 x_{13}^3 x_{23}} \,, \\
        \iota^*f_1 &= \frac{(x_{13} x_{22}-x_{12} x_{23}) (2 a_2 x_{13} (a_1 x_{23} x_{33}+a_3 x_{13} x_{21})+a_1 a_3 x_{23} (x_{12} x_{23}+x_{13}x_{22}))}{a_2 x_{13}^2 x_{23}} \,, \\
        \iota^*f_2 &= a_1 a_2 x_{33} +2 a_1 a_3 x_{22}-\frac{a_1 a_3 x_{12} x_{23}}{x_{13}}+\frac{a_2a_3 x_{13} x_{21}}{x_{23}} \,.
    \end{split}
\end{equation}
Here we define $\imath$ to be the inclusion map $\mathbb{F}^6\hookrightarrow \mathfrak{gl}_3(\mathbb{F})$ 
$$\iota:(x_{12},x_{13},x_{21},x_{22},x_{23},x_{33}) \mapsto X=(x_{ij}), $$ with $x_{11},x_{31},x_{32}$ given by \eqref{x11x31x32}.

In what follows, for simplicity we will denote all $\iota^*f_i$ by $\widetilde{f}_i$. We also denote $\mathcal{M}=\text{Img}(\iota)$, i.e. the set of reduced matrices in $\mathcal{L}$, with $x_{11}, x_{31},x_{32}$ defined by \eqref{x11x31x32}. The following proposition holds:
\begin{proposition}\label{prop: inclus}
$\mathcal{M}$ is a Poisson submanifold of $\mathcal{L}$ of rank four. Furthermore, the discriminant of the cubic polynomial in $\lambda$ of $\iota^*p^a_{\lambda}$ vanishes, i.e.
 \begin{equation}
4\widetilde{f}_0\widetilde{f}_2^3-\widetilde{f}_1^2 \widetilde{f}_2^2+4 \widetilde{f}_3\widetilde{f}_1^3-18\widetilde{f}_0\widetilde{f}_1\widetilde{f}_2\widetilde{f}_3+27 \widetilde{f}_3^2\widetilde{f}_0^2=0.
\label{discr}
\end{equation}
\end{proposition}
\begin{proof}
By direct computation we can show that for $x_{13}, x_{23} \neq 0$, the inclusion map $\iota$ defined by \eqref{x11x31x32} is Poisson with respect to the Poisson bracket \eqref{PBs} (and the induced bracket on $\mathcal{M}$).  Hence, $\mathcal{M}$ is a Poisson submanifold of $\mathcal{L}$ and by substituting \eqref{x11x31x32} in \eqref{Skly} reduces the rank of the Poisson structure matrix to four. The coefficients of the pullback $\iota^*p^a_{\lambda}$ are the reduced Casimirs $\iota^*f_i:=\widetilde{f}_i$ given in \eqref{casM}. The rank of the Jacobian matrix of $\widetilde{f}_0,\widetilde{f}_1,\widetilde{f}_2$ is equal to two, and therefore there is one functional relation between them. Relation \eqref{discr} can be obtained using elimination algorithms.
\end{proof}

As we mentioned in the proof of Proposition \eqref{prop: inclus}, on the submanifold $\mathcal{M}$ there are two functionally independent Casimirs. Solving the system $\widetilde{f}_1 = \alpha_1, \widetilde{f}_2 = \alpha_2$ for $x_{22}$ and $x_{33}$ we obtain the following expressions 
\begin{equation}\label{x22x33}
    \begin{split}
        x_{22} &= \frac{c_1}{a_3} + \frac{c_2}{a_3} + \frac{x_{12} x_{23}}{x_{13}}\,, \\
        x_{33} &= \frac{c_1}{a_2} - \frac{2 c_2}{a_2} -\frac{a_3 x_{13} x_{21}}{a_1 x_{23}} - \frac{a_3 x_{12}x_{23}}{a_2 x_{13}} \,,
    \end{split}
\end{equation}
where $c_1$ and $c_2$ depend on the level sets of the Casimirs and the matrix $K_a$ as follows:
\begin{equation} \label{c1c2}
    c_1 = \frac{\alpha_2}{3a_1}\,, \quad c_2 = \pm \frac{\sqrt{\alpha_2^2-3\alpha_1 a_1a_2a_3}}{3a_1}\,.
\end{equation}
Using relations \eqref{x22x33} in \eqref{x11x31x32} we obtain the reduced expressions for $x_{11},~x_{31},~x_{32}$
\begin{equation}\label{x11x31x32 red}
    \begin{split}
        x_{11} &= \frac{a_1 (c_1 + c_2)}{a_2 a_3} + \frac{x_{13} x_{21}}{x_{23}}  \,,\\
        x_{31} &= - \frac{3 c_2 x_{21}}{a_2 x_{23}} -\frac{a_3 x_{13} x_{21}^2}{a_1 x_{23}^2} - \frac{a_3 x_{12} x_{21}}{a_2 x_{13}}  \,,\\
        x_{32} &= - \frac{3 c_2 x_{12}}{a_2 x_{13}} -\frac{a_3 x_{12} x_{21}}{a_1 x_{23}}-\frac{a_3 x_{12}^2 x_{23}}{a_2 x_{13}^2}\,.
    \end{split}
\end{equation}
After the reduction, the Poisson brackets \eqref{PBs} between the remaining variables $x_{12},x_{13},x_{21}, x_{23}$ take the form
\begin{equation}\label{PB red1}
   \lbrace  x_{13}, x_{21} \rbrace = a_1 x_{23}  \,, \quad \lbrace x_{12}, x_{21} \rbrace = a_1 \frac{x_{12}x_{23}}{x_{13}} - a_2 \frac{x_{13} x_{21}}{x_{23}}  \,, \quad \lbrace x_{12}, x_{23} \rbrace = -a_2 x_{13} \,,
\end{equation}
with all other brackets vanishing.  Similarly, after the reduction \eqref{x22x33}-\eqref{x11x31x32 red} the Lax matrix \eqref{eq: Lax} takes the form $\tilde{X} -\lambda K_a$ where $\tilde{X}$ is given by
 \begin{equation*}
    \begin{pmatrix}
       \frac{a_1 (c_1 + c_2)}{a_2 a_3}+\frac{x_{13} x_{21}}{x_{23}} & x_{12} & x_{13} \\
 x_{21} & \frac{c_1 + c_2}{a_3} + \frac{x_{12} x_{23}}{x_{13}} & x_{23} \\
- \frac{3 c_2 x_{21}}{a_2 x_{23}} -\frac{a_3 x_{13} x_{21}^2}{a_1 x_{23}^2} - \frac{a_3 x_{12} x_{21}}{a_2 x_{13}} & -  \frac{3 c_2 x_{12}}{a_2 x_{13}}  -\frac{a_3 x_{12} x_{21}}{a_1 x_{23}} - \frac{a_3 x_{12}^2 x_{23}}{a_2 x_{13}^2}   & \frac{c_1 -  2c_2}{a_2} -\frac{a_3 x_{13}x_{21}}{a_1 x_{23}} - \frac{a_3 x_{12} x_{23}}{a_2 x_{13}}  \\ 
     \end{pmatrix}.
\end{equation*}

The change of variables
\begin{equation}\label{can vars}
    x_{12} = -a_2 x_2X_1 \,, \quad x_{13} = X_1 \,, \quad x_{21} = -a_1x_1X_2 \,, \quad  x_{23} = X_2\,,
\end{equation}
brings the brackets \eqref{PB red1} to the canonical form, i.e.
\begin{equation}\label{PBred can}
    \lbrace x_1, X_1 \rbrace = 1\,, \quad \lbrace x_2, X_2 \rbrace = 1\,, \quad \lbrace x_1, x_2 \rbrace = 0\,, \quad \lbrace X_1, X_2 \rbrace = 0 \,,
\end{equation}
while the reduced Lax matrix $\tilde{X} -\lambda K_a$ takes the form
\begin{flalign} \label{Lax red can}
  L({\bm x}, {\bm X}, p, \lambda)  = &&
  \end{flalign}
 {\small
 \begin{equation*}
\begin{pmatrix}
 \frac{a_1 (c_1 + c_2)}{a_2 a_3} - a_1x_1X_1 - \lambda a_1  & -a_2 x_{2}X_1 & X_1 \\
 -a_1x_1X_2 & \frac{c_1 + c_2}{a_3} -a_2x_2X_2 - \lambda a_2 & X_2 \\
 \frac{3 a_1 c_2 x_1}{a_2} - a_1a_3x_1(x_1X_1+x_2X_2) &  3c_2x_2 -a_2a_3x_2(x_1X_1+x_2X_2)   & \frac{c_1 - 2c_2}{a_2} +a_3(x_1X_1+x_2X_2) - \lambda a_3  \\
    \end{pmatrix},
 \end{equation*}   
}
where ${\bm x}= (x_1, x_2)$, ${\bm X}=(X_1, X_2)$ and $p = (c_1,c_2, a_1,a_2,a_3)$.

The matrix refactorisation problem
\begin{equation}
  L({\bm u}, {\bm U},p,\lambda) L({\bm v}, {\bm V}, q,\lambda) = L({\bm y}, {\bm Y}, q,\lambda) L({\bm x}, {\bm X},p,\lambda),
\end{equation}
with $q=(d_1,d_2,b_1,b_2,b_3)$ has a unique solution for ${\bm u}, {\bm U}, {\bm v}, {\bm V}$ in terms of ${\bm x}, {\bm X}, {\bm y}, {\bm Y}$ given by
\begin{equation}\label{map u}
\begin{split}
    \left(u_1, u_2 \right) &= \frac{b_3}{a_3} \left(y_1, y_2 \right) + \frac{b_3}{a_3} \frac{C_1}{D_1} \left( \frac{a_3}{b_1}x_1-y_1, \frac{a_3}{b_2} x_2 -y_2 \right),\\
\left( v_1, v_2 \right) &= \left(\frac{a_1}{b_1} x_1, \frac{a_2}{b_2} x_2   \right) - b_1 \frac{C_2}{D_2} \left( a_1 ( \frac{a_3}{b_1}x_1 -y_1 ), a_2 ( \frac{a_3}{b_2}x_2 - y_2 )  \right),
    \end{split}
\end{equation}
where
\begin{equation}
    \begin{split}
      C_1 &=  a_2a_3b_1 (d_1-2d_2) - b_1b_2b_3(c_1 - 2c_2),\\  
      C_2 &= \frac{b_2b_3}{a_3}(c_1+c_2) - a_2(d_1+d_2), \\
      D_1 & = a_2a_3b_2b_3 (a_3x_1-b_1y_1)Y_1 + a_2a_3b_1b_3 (a_3x_2-b_2y_2)Y_2 + C_1 + 3a_2a_3b_1d_2, \\
      D_2 & = a_2b_1b_2b_3(a_3x_1-b_1y_1)X_1 + a_2b_1b_2b_3(a_3x_2-b_2y_2)X_2 -C_1 -3a_2a_3b_1d_2,
    \end{split}
    \label{CDs}
\end{equation}
and 
\begin{equation} \label{map U}
    \begin{split}
        U_1 &= \frac{(a_1x_1 - b_1v_1)X_1 + (a_1y_1-a_3v_1)Y_1}{a_1u_1 - b_3v_1}, \; U_2 = \frac{(a_2x_2-b_2v_2)X_2 + (a_2y_2 - a_3v_2)Y_2}{a_2u_2 - b_3v_2}, \\
        V_1 &= \frac{(b_1u_1 - b_3x_1)X_1 + (a_3u_1 - b_3y_1)Y_1}{a_1u_1 - b_3v_1}, \; V_2 = \frac{(b_2u_2 - b_3x_2)X_2 + (a_3u_2 - b_3y_2)Y_2}{a_2u_2 - b_3v_2}\,.
    \end{split}
\end{equation}

Finally, by direct computation we can prove the following proposition:
\begin{proposition} 
The map 
\begin{equation} \label{YB2nd}
R_{p,q}:((x_1,x_2,X_1,X_2),(y_1,y_2,Y_1,Y_2))\mapsto ((u_1,u_2,U_1,U_2),(v_1,v_2,V_1,V_2)) \;,
\end{equation}
with $u_i,U_i,v_i,V_i$ defined by \eqref{map u}-\eqref{map U} is a  parametric quadrirational Yang-Baxter map with strong Lax  matrix \eqref{Lax red can}. Furthermore, $R_{p,q}$ is symplectic with respect to 
$$\omega= dx_1 \wedge dX_1 + dx_2 \wedge dX_2+ dy_1 \wedge dY_1+ dy_2 \wedge dY_2\;.$$
\end{proposition}

The map \eqref{YB2nd} admits four functionally independent invariants which can be obtained from the characteristic polynomial of the monodromy matrix $L({\bm y},{\bm Y},q,\lambda)L({\bm x},{\bm X},p,\lambda)$. The two simplest invariants obtained in this way are  
\begin{equation} 
\label{I1I2}
   I_1 = x_1X_1  + y_1Y_1 \,, \quad I_2 = x_2X_2 + y_2Y_2.
\end{equation}

\begin{remark}
For $a_i=b_i=1$ and particular choices of the parameters $c_i$, $d_i$, the YB map \eqref{YB2nd} reduces to  non-degenerate Boussinesq and Goncharenko-Veselov maps \cite{GonVes,KP2011}.
\end{remark}

\subsection{Reduction to a $4$-dimensional Yang-Baxter map}\label{sec:reduction}

In this section, we consider a further folding reduction of the Yang-Baxter map \eqref{YB2nd} on an invariant manifold. We impose the relations
\begin{equation} \label{red 2}
\begin{split}
    &x_1 = x_2 :=x \,, \quad X_1 = X_2:= X \,, \quad a_1=a_2:=a\;, \\
     & y_1 = y_2 :=y \,, \quad Y_1 = Y_2:= Y \,, \quad b_1=b_2:=b\;,
     \end{split}
\end{equation}
which are consistent with the map since relations \eqref{map u}-\eqref{map U} imply $$u_1=u_2, \ U_1=U_2, \ v_1=v_2, \ V_1=V_2.$$
Thus, for $p=(c_1,c_2,a,a_3)$ and $q=(d_1,d_2,b,b_3)$,  the map \eqref{YB2nd} is reduced to a $4$-dimensional quadrirational YB map $$\mathbb{R}_{p,q}:((x,X),(y,Y))\mapsto((u,U),(v,V))$$ on the invariant manifold
$$\mathcal{N}=\{((x,x,X,X),(y,y,Y,Y)):x,X,y,Y \in \mathbb{F}\}\,,$$
where 
\begin{equation}\label{YB red2}
\begin{split}
u = \frac{b_3}{a_3}y+\frac{b_3}{a_3} \frac{\widetilde{C}_1}{\widetilde{D}_1}\left(\frac{a_3}{b}x-y \right), & \quad v = \frac{a}{b}x-a b\frac{\widetilde{C}_2}{\widetilde{D}_2}\left(\frac{a_3}{b}x-y \right),
  \\ 
  U=\frac{(a x-b v)X+(ay-a_3 v)Y}{a u-b_3 v}, & \quad V=\frac{(b u-b_3 x)X+(a_3u-b_3 y)Y}{a u-b_3 v},
\end{split}    
\end{equation}
and $\widetilde{C}_i,~\widetilde{D}_i$ are obtained from \eqref{CDs} under the reduction \eqref{red 2}.

The reduced symplectic structure on $\mathcal{N}$ is 
$$\bbomega=dx\wedge dX + dy \wedge dY\;,$$
and $\mathbb{R}_{p,q}$ is symplectic with respect to $\bbomega$. Furthermore, map $\mathbb{R}_{p,q}$ admits the reduced Lax matrix
\begin{equation} \label{4d Lax}
\mathbb{L}(x, X, p, \lambda)  =
\begin{pmatrix}
 \frac{c_1 + c_2}{a_3} - a x X - \lambda a  & -a x X  & X \\
 -a x X & \frac{c_1 + c_2}{a_3} -a x X - \lambda a & X \\
 3c_2 x - 2aa_3 x^2 X &  3c_2x -2aa_3 x^2 X   & \frac{c_1 - 2c_2}{a} + 2 a_3 x X - \lambda a_3 
    \end{pmatrix},
 \end{equation}  
which can be obtained from \eqref{Lax red can} by imposing the reduction \eqref{red 2}.

Using the trace of the monodromy matrix associated to Lax matrix \eqref{4d Lax} we obtain two functionally independent invariants of the map \eqref{YB red2}. One of the invariants is $\mathbb{I}_1=xX+yY$, which is $I_1$ (or $I_2$) under the reduction \eqref{red 2}, while the other has the form
\begin{equation}
\mathbb{I}_2=a_{11}^{00}~xX+a_{00}^{11}~yY+a_{10}^{01}~xY+a_{01}^{10}~Xy+XY(a_{21}^{01}x^2+a_{01}^{21}y^2+a_{11}^{11}xy),
\end{equation}
where the coefficients $a_{ij}^{kl}$ of the monomials $x^iX^jy^kY^l$ depend on the parameters of the map \eqref{YB red2}. For simplicity, we have omitted the exact dependence of the coefficients $a_{ij}^{kl}$ on the parameters of the map. Moreover, the invariants $\mathbb{I}_1$ and $\mathbb{I}_2$ Poisson commute with respect to the canonical Poisson structure $\lbrace x,X \rbrace=1$ and $\lbrace y,Y \rbrace=1$ and therefore the map \eqref{YB red2} is a 4-dimensional symplectic quadrirational map that is also integrable in the Liouville sense.

\section{Degenerate Limits}

\subsection{A single degenerate limit}

We are interested in studying zero limits for certain parameters involved in the maps derived in Section \ref{sec:3x3}, which effectively result in degenerations for the maps introduced in \cite{KP2011}. We call such limits degenerate. In particular, in this section we focus on the degenerate limit $a_3 \to 0$.

We consider the limit $a_3 \to 0$ of expressions \eqref{x22x33}. For both values of $c_2$ given in \eqref{c1c2}, the level set $\alpha_2$ of the Casimir $f_2$ can be chosen so that the branch of the square root is such that the limit $a_3 \to 0$ of $x_{22}$ results in a unique well-defined expression. The same limit of $x_{33}$ is regular. The obtained expressions are the following 
\begin{equation} \label{x22x33 a3}
        \lim_{a_3 \rightarrow 0} x_{22} =  \frac{a_2 \alpha_1}{2\alpha_2}+\frac{x_{12} x_{23}}{x_{13}} \,, \qquad
        \lim_{a_3 \rightarrow 0} x_{33} = \frac{\alpha_2}{a_1 a_2} \,.
\end{equation}
Substituting formulas \eqref{x22x33 a3} in the expressions  \eqref{x11x31x32} for $x_{11}, x_{31}, x_{32}$ and taking the now regular limit $a_3 \rightarrow 0$ we obtain
\begin{equation}
        x_{11} = \frac{a_1\alpha_1}{2\alpha_2}+\frac{x_{13} x_{21}}{x_{23}} \,, \quad 
        x_{31} = \frac{\alpha_2 x_{21}}{a_1 a_2 x_{23}} \,,\quad
        x_{32} = \frac{\alpha_2 x_{12}}{a_1 a_2 x_{13}} \,.
\end{equation}
Hence, in the limit $a_3 \rightarrow 0$ the Lax matrix $\tilde{X}-\lambda K_a$ takes the form
\begin{equation} \label{Lax degen1}
\lim_{a_3 \to 0} \left( \tilde{X}-\lambda K_a \right)=  \begin{pmatrix}
 \frac{a_1 \alpha_1}{2 \alpha_2} + \frac{x_{13} x_{21}}{x_{23}} -\lambda a_1  & x_{12} & x_{13} \\
 x_{21} & \frac{a_2 \alpha_1}{2 \alpha_2} +\frac{x_{12} x_{23}}{x_{13}} - \lambda a_2  & x_{23} \\
 \frac{\alpha_2 x_{21}}{a_1 a_2 x_{23}} & \frac{\alpha_2 x_{12}}{a_1 a_2 x_{13}} & \frac{\alpha_2}{a_1 a_2} \\
    \end{pmatrix}.
\end{equation}

The change of variables \eqref{can vars} is not affected by the limit $a_3 \rightarrow 0$, and in the variables $x_1, x_2, X_1, X_2$ the Lax matrix \eqref{Lax degen1} takes the form
\begin{equation} \label{Lax degen1 can}
\widetilde{L}({\bm x}, {\bm X}, p,\lambda) =   \begin{pmatrix}
 \frac{a_1 \alpha_1}{2\alpha_2} - a_1x_1X_1 - \lambda a_1  & -a_2 x_2 X_1  & X_1 \\
 -a_1x_1X_2 & \frac{a_2 \alpha_1}{2 \alpha_2} - a_2x_2X_2 - \lambda a_2 & X_2 \\
 -\frac{\alpha_2}{a_2} x_1 & -\frac{\alpha_2}{a_1} x_2 & \frac{\alpha_2}{a_1a_2}  \\
    \end{pmatrix},
\end{equation}
with $p = (\alpha_1, \alpha_2,a_1,a_2)$. The refactorisation problem associated to the Lax matrix \eqref{Lax degen1 can} results in the following birational (but non quadrirational) YB map $\widetilde{R}_{p,q}: (x_1,x_2,X_1,X_2,y_1,y_2,Y_1,Y_2)\to(u_1,u_2,U_1,U_2,v_1,v_2,V_1,V_2)$ where
\begin{equation}
\label{eq: degen lim YB 1}
    \begin{split}
        u_i &= \frac{a_1 a_2 \beta_2}{b_1 b_2 \alpha_2}~ y_i\,, \\
        U_i &= \frac{b_1 b_2 \alpha_2}{a_1 a_2 \beta_2}~ Y_i -\frac{b_1b_2b_i k}{a_1a_2 \beta_2(b_1y_1X_1 + b_2y_2X_2) - \alpha_2 \beta_2} ~ X_i \,, \\
        v_i &= \frac{a_i}{b_i} ~ x_i + \frac{a_1 a_2 a_i k}{a_1a_2 \alpha_2(b_1y_1X_1 + b_2y_2X_2) - \alpha_2^2} ~ y_i \\
        V_i &= \frac{b_i}{a_i}~ X_i \,,
    \end{split}
\end{equation}
for $i=1,2$ and with $k=\frac{\alpha_2 \beta_1 -\alpha_1 \beta_2}{2 \beta_2}$. This 8-dimensional YB map admits four functionally independent polynomial invariants. Two of these are given in \eqref{I1I2}, which remain unaffected by the limit $a_3 \to 0$. Hence, we have $\widetilde{I}_1 = x_1X_1  + y_1Y_1$, $\widetilde{I}_2 = x_2X_2 + y_2Y_2$. Additionally, from the spectrum of the monodromy matrix, we obtain the following two invariants:
\begin{equation}
\begin{split}
    \widetilde{I}_3 &= -a_1^2 a_2 b_1^2 b_2 ( \alpha_2 \beta_1 x_1X_1 + \alpha_1 \beta_2 y_1Y_1 ) -a_1 a_2^2b_1b_2^2 (\alpha_2 \beta_1 x_2X_2 + \alpha_1 \beta_2 y_2Y_2 ) \\
   & + 2\alpha_2\beta_2 \Bigl( b_1b_2 (a_1x_1Y_1 +a_2x_2Y_2) -\beta_2  \Bigr) \Bigl( a_1a_2 ( b_1y_1X_1 +b_2y_2X_2 ) - \alpha_2 \Bigr),
    \end{split}
\end{equation}
and
\begin{equation}
    \begin{split}
        \widetilde{I}_4 & = a_1a_2b_1b_2 \Bigl( \alpha_2\beta_1 (x_1X_1 + x_2X_2) + \alpha_1\beta_2 (y_1Y_1+y_2Y_2)  \Bigr)  + 2\alpha_2^2\beta_2 ( b_2x_1Y_1 +b_1x_2Y_2 ) \\
        &  + 2\alpha_2\beta_2^2 ( a_2y_1X_1 + a_1y_2X_2 ) -2a_1a_2b_1b_2\alpha_2\beta_2 (x_1X_1 +x_2X_2) (y_1Y_1 + y_2Y_2).
    \end{split}
\end{equation}

The folding reduction \eqref{red 2} can also be applied to the map \eqref{eq: degen lim YB 1}. Indeed, in this case the map will simplify to the four dimensional YB map $\widetilde{\mathbb{R}}_{a,b}: (x,X,y,Y) \to (u,U,v,V)$ given by
\begin{equation}
\label{eq: reduced degen lim YB 1}
    \begin{split}
        u = \frac{a^2 \beta_2}{b^2 \alpha_2} y, & \quad U = \frac{b^2 \alpha_2}{a^2 \beta_2}Y -\frac{b^3 k}{2a^2b \beta_2 yX  - \alpha_2 \beta_2} X , \\
       V = \frac{b}{a} X \,,  & \quad  v = \frac{a}{b}  x + \frac{a^3 k}{2a^2b \alpha_2 yX - \alpha_2^2}  y\,. 
    \end{split}
\end{equation}
The invariants of map \eqref{eq: degen lim YB 1} will reduce to invariants of map $\widetilde{\mathbb{R}}_{a,b}$ given in \eqref{eq: reduced degen lim YB 1}. Under the folding reduction, both invariants $\widetilde{I}_1$, $\widetilde{I}_2$  will become equal to $\widetilde{\mathbb{I}}=x X+y Y$, and similarly both invariants $\widetilde{I}_3$ and $\widetilde{I}_4$ will take the same form (up to a scaling factor), resulting in the following invariant
\begin{equation}
\widetilde{\mathbb{J}}= a^2b^2(\alpha_2\beta_1x X+\alpha_1\beta_2yY)+2\alpha_2\beta_2(b\alpha_2xY+a\beta_2yX)-4a^2b^2\alpha_2\beta_2x X y Y.
    \label{invariant reduced degen J}
\end{equation}
Both invariants $\widetilde{\mathbb{I}}$, $\widetilde{\mathbb{J}}$ can be obtained from the trace of the monodromy associated to the Lax matrix \eqref{Lax degen1 can} after applying the reduction \eqref{red 2}. The invariants $\widetilde{\mathbb{I}}$ and $\widetilde{\mathbb{J}}$ Poisson commute therefore the map \eqref{eq: reduced degen lim YB 1} is symplectic and Liouville integrable.

\begin{remark}
    The 4-dimensional map \eqref{eq: reduced degen lim YB 1} can also be obtained from the refactorisation of the Lax \eqref{4d Lax} by taking the limit $a_3\to 0$.
\end{remark}

\begin{remark}
The choice of parameters
$$
    a=b=1,\; \alpha_2=\beta_2=-1, \; \alpha_1 = 2A, \; \beta_1 = 2 B,
$$
results in the map
\begin{equation}
\label{eq: rescaled AY}
u =  y, \quad U = Y -\frac{A-B}{1+2yX} X, \quad v = x + \frac{A-B}{1+2yX}  y, \quad V =  X, 
\end{equation}
which after the rescaling of the variables $x,y,u,v$ by a factor of $\frac{1}{2}$ and the flip $x\leftrightarrow X$, $y\leftrightarrow Y$, $u\leftrightarrow U$ and $v \leftrightarrow V$ becomes the well-known AY map \cite{AdlerYamilov1994}.  
\end{remark}

\subsection{A double limit to a vectorial Adler-Yamilov map}

In this section we aim to demonstrate that a vectorial generalisation of the Adler-Yamilov map (vAY) together with its standard Lax matrix, can be obtained from the Lax matrices that we study after taking the double limit $a_2 \rightarrow 0$ and $a_3 \rightarrow 0$ in $K_a$. Since in this case the only nonzero parameter of $K_a$ is $a_1$, without loss of generality we can rescale it to one. Therefore, in this section we have that $K_a=e_{11}$. 

The double limit $(a_2,a_3)\to (0,0)$ of the Casimir $f_2$ in \eqref{Cas} exists and implies that $f_2=\alpha_2\equiv 0$. This means that the iterated single limits commute. Taking first the limit of $f_2$ when $a_2 \to 0$ we obtain that $a_3x_{22}=\alpha_2$, while if we take first the limit $a_3\to 0$ it follows that $a_2x_{33}=\alpha_2$. Hence, when $a_2,a_3 \to 0$, both $x_{22}$ and $x_{33}$ have to be constants. The case $x_{22}=x_{33}=0$ leads to a trivial refactorisation problem, therefore we consider the case $x_{22}=x_{33}=1$ which occurs only when $a_2=a_3=\epsilon\to 0$. If we restrict on the level set $\alpha_1=1$ of the Casimir $f_1$, then we obtain the constraint $x_{23}x_{32}=0$. We assume that $x_{23}=x_{32}=0$. Finally, under the above assumptions, the Casimir $f_0$ leads to the following constraint  
\begin{equation}
f_0:=x_{11}-x_{12}x_{21}-x_{13}x_{31}= a,
    \label{x11 a2a3 red}
\end{equation}
which we use to solve for $x_{11}$. 

Taking into account all the above, the reduced Lax matrix is of the form
\begin{equation}
\mathfrak{L} ({\bm x}, {\bm X},a,\lambda)=
\begin{pmatrix}
    a-\lambda+x_{1}X_{1}+x_{2}X_{2} & x_{1} & x_{2} \\
    X_{1} & 1 & 0 \\
    X_{2} & 0 & 1
\end{pmatrix},
    \label{eq: AY Lax}
\end{equation}
where we use the notation $x_{1j}=x_{j-1}$ and $x_{i1}=X_{i-1}$ for $i,j=2,3$. The refactorisation problem associated to the Lax matrix \eqref{eq: AY Lax} implies uniquely the non-quadritational YB map $\mathfrak{R}_{a,b}: ({\bm x}, {\bm X}, {\bm y}, {\bm Y}) \to ({\bm u},{\bm U},{\bm v},{\bm V})$
\begin{equation}
    \begin{split}
        {\bm u}={\bm y}-\frac{a-b}{1+\left<{\bm x},{\bm Y}\right>}{\bm x}, & \quad{\bm U}={\bm Y},\\
        {\bm V}={\bm X}+\frac{a-b}{1+\left<{{\bm x}},{\bm Y} \right>}{\bm Y}, & \quad {\bm v}={\bm x},
    \end{split}
    \label{eq AY map}
\end{equation}
where $\left<\cdot,\cdot\right>$ is the standard bilinear form in $\mathbb{F}^2$. From the spectrum of the monodromy matrix we obtain the usual invariants \eqref{I1I2} 
\begin{equation}
    \mathfrak{I}_1=x_1 X_1+y_1 Y_1\,, \quad \mathfrak{I}_2=x_2 X_2+y_2 Y_2\,,
\end{equation}
together with
\begin{equation}
    \mathfrak{I}_3 = b\left<{\bm x},{\bm X}\right>+a\left<{\bm y},{\bm Y}\right>+\left<{\bm x},{\bm Y}\right>+\left<{\bm X},{\bm y}\right> + \left<{\bm x},{\bm X}\right>\left<{\bm y},{\bm Y}\right>.
\end{equation}
Moreover, it can be verified by a direct calculation that the determinants 
$x_1y_2-x_2y_1$, $X_1Y_2-X_2Y_1$ are anti-invariants, and therefore their product
\begin{equation}
    \mathfrak{I}_4=(x_1y_2-x_2y_1)(X_1Y_2-X_2Y_1)
\end{equation}
is another invariant of the map \eqref{eq AY map}.

\begin{proposition} \label{prop: integrab vector AY}
The YB map $\mathfrak{R}_{a,b}$ in \eqref{eq AY map} is Liouville integrable and symplectic with respect to the canonical symplectic structure.
\end{proposition}
\begin{proof}
The Jacobian matrix of the invariants $\mathfrak{I}_1,\ldots, \mathfrak{I}_4$ has full rank and therefore the invariants are functionally independent. Moreover, the invariants commute with respect to the Poisson bracket 
$$
 \{ x_i, X_i \} = \delta_{ij}, \quad \{y_i, Y_i \} = \delta_{ij}
$$
 implied by the Sklyanin bracket \eqref{Skly}. The map \eqref{eq AY map} preserves the Poisson bracket and therefore is symplectic with respect to the corresponding canonical symplectic structure
$$
\omega=\sum_{i=1}^{2}dx_i\wedge dX_i+dy_i\wedge dY_i.
$$
\end{proof}
 
\begin{remark}
    The YB map $\mathfrak{R}_{a,b}$ in \eqref{eq AY map} together with its Lax matrix \eqref{eq: AY Lax} has an obvious generalisation to $n$-dimensional $\mathbb{F}$-vectors. The integrability properties of the map can be generalised to arbitrary dimensions in a similar way as the one described in Proposition \eqref{prop: integrab vector AY}. The $n$-dimensional map has been derived in \cite{KRM2013}.
\end{remark}

The reduction \eqref{red 2} to the invariant manifold $\mathcal{N}$ presented in Section \eqref{sec:reduction} can also be applied in this case, and reduces the map $\mathfrak{R}_{a,b}$ in \eqref{eq AY map} to the standard AY map \cite{AdlerYamilov1994}. Moreover, map \eqref{eq AY map} can be obtained from the YB map $\widetilde{R}_{p,q}$ in \eqref{eq: degen lim YB 1} after the following choice for the parameters
$$
a_1= b_1 = 1, \: a_2 = b_2 = 1, \:\alpha_2 = \beta_2 = -1, \: \alpha_1 = 2a, \: \beta_1 = 2 b\,,
$$
and the permutation in the variables
$$
{\bm x} \leftrightarrow {\bm X}, \, {\bm y} \leftrightarrow {\bm Y}, \, {\bm u} \leftrightarrow {\bm U}, \, {\bm v} \leftrightarrow {\bm V} \,.
$$
This implies that the vAY map \eqref{eq AY map} admits the Lax matrix
\begin{equation} \label{AY Lax 2}
\widetilde{L}({\bm X}, {\bm x}, a) =  \begin{pmatrix}
      x_1X_1 +a + \lambda & x_1X_2 & -x_1 \\
      x_2X_1 & x_2X_2 + a + \lambda & -x_2 \\
      -X_1 & -X_2 & 1
  \end{pmatrix}.
\end{equation}
This Lax matrix is equivalent to the Lax matrix \eqref{eq: AY Lax} as it can be directly obtained from \eqref{eq: AY Lax} by inversion and a similarity transformation by a permutation matrix.

\section{Conclusions}

In this paper we studied reductions and degenerations of a family of YB maps with $3 \times 3$ first-degree polynomial Lax matrices. We also introduced compatible Poisson structures associated with the Sklyanin bracket and invariant conditions, ensuring the Liouville integrability of all the presented maps.

Yang-Baxter maps serve as  fundamental building blocks for constructing higher-dimensional discrete integrable systems. As it was shown in \cite{Veselov2003, Veselov2007}, each YB map generates an hierarchy of commuting $n$-dimensional {\it  transfer maps} which share the same integrals. A different variant of transfer maps occurs by considering periodic initial value problems on lattices, in correspondence with staircase initial value problems of integrable lattice equations \cite{PNC1990, QCPN1991}. While the dynamics of an individual YB map sometimes may be trivial (e.g. involution), the dynamics of the associated transfer maps display highly non-trivial behavior.

All variations of transfer maps preserve the spectrum of their corresponding monodromy matrices, constituting of products of Lax matrices. In this setting, YB maps with Lax matrices compatible with the Sklyanin bracket, as the maps presented in this paper, hold a significant advantage as they generate Poisson transfer maps with commutative integrals derived from the spectrum of the corresponding monodromy matrices.

The family of the YB maps presented in this work includes all cases involving $3\times 3$ Lax matrices linear in the spectral parameter and with the degree one coefficient being a constant and diagonalisable matrix. The connection between the maps that we obtained can be summarised in the graph below. The arrows in the graph indicate a connection between the maps which can be a reduction, a limit, a choice of parameters, or a simple transformation in the dynamical variables that is a symmetry of the YB equation.

In future research, we aim to study families of YB maps and corresponding degenarations, with $3 \times 3$ Lax matrices associated with different Jordan forms of the highest-degree terms. This will complete the classification of $3 \times 3$ binomial Lax matrices under conjugation. Furthermore, we intend to include the $3$-dimensional consistent lattice equations associated with all these maps and investigate their integrability features. 

\begin{figure}[ht!]
\centering
\begin{small}
\begin{tikzpicture}[node distance=1.5cm and 2.5cm, auto]
\node (maps_title) at (0,0) {\textbf{maps}};
\node (dim_title) [right=of maps_title, xshift=4.5cm] {\textbf{dimensions}};
\node (R_ab) [below=of maps_title] {$\mathbf{R}_{ab}$};
\node (R_a'b') [below=of R_ab] {$\mathrm{R}_{\bar{a},\bar{b}}$};
\node (R_pq) [below=of R_a'b'] {$R_{p,q}$};
\node (R_pq_2) [below=of R_pq, yshift=-0.2cm] {$\mathbb{R}_{p,q}$}; 
\node (R_pq_tilde) [right=of R_pq, xshift=0.1cm] {$\widetilde{R}_{p,q}$};
\node (R_pq_tilde1) [right=of R_pq_tilde, xshift=0.1cm] {vAY};
\node (R_pq_tilde_2) [below=of R_pq_tilde, yshift=-0.1cm] {$\widetilde{\mathbb{R}}_{p,q}$};
\node (R_pq_tilde2) [right=of R_pq_tilde_2, xshift=0.25cm] {AY};
\node (dim18) [right=of R_ab, xshift=5.4cm] {18};
\node (dim12) [right=of R_a'b', xshift=5.4cm] {12};
\node (dim6) [right=of R_pq, xshift=5.5cm] {8};
\node (dim4) [right=of R_pq_2, xshift=5.5cm] {4};
\draw[->] (R_ab) -- (R_a'b');
\draw[->] (R_a'b') -- (R_pq);
\draw[->] (R_pq)   -- ++ (0, -1.8) node[midway, left] {} -- (R_pq_2); 
\draw[->] (R_pq_tilde1)  -- ++ (0, -1.9) node[midway, left] {} -- (R_pq_tilde2);
\draw[->] (R_pq) -- ++(2.8, 0) node[midway, above] {$a_3\to 0$} -- (R_pq_tilde);
\draw[->] (R_pq_2) -- ++(3.1, 0) node[midway, above] {$a_3\to 0$} -- (R_pq_tilde_2);
\draw[->] (R_pq_tilde) -- ++(3, 0) node[midway, above] {} -- (R_pq_tilde1);
\draw[->] (R_pq_tilde_2) -- ++(3, 0) node[midway, above] {} -- (R_pq_tilde2);
\draw[->] (R_pq_tilde) -- ++ (0, -1.8) node[midway, left] {} -- (R_pq_tilde_2);
\end{tikzpicture}
\end{small}
\end{figure}

\end{document}